\documentclass[12pt]{article}

\usepackage{hyperref,graphicx,amsthm,amsmath,amsfonts,color,algorithm,fullpage}

\newtheorem{theorem}{Theorem}[section]
\newtheorem{corollary}[theorem]{Corollary}

\newtheorem{lemma}[theorem]{Lemma}

\newcommand{\1}{\mathbf 1}

\newcommand{\F}{\mathrm F}
\newcommand{\R}{\mathbb R}

\title{A Note on Projection-Based Recovery of Clusters in Markov Chains}
\author{Sam Cole\\ 
University of Missouri, Department of Mathematics\\
\href{mailto:s.cole@missouri.edu}{\texttt{s.cole@missouri.edu}}}

\begin{document}

\maketitle

\abstract{
In this companion work to~\cite{cole2021clusters}, we discuss identification of clusters in Markov chains.  Let $T_0$ be the transition matrix of a \emph{purely clustered} Markov chain, i.e.\ a direct sum of $k \geq 2$ irreducible stochastic matrices.  Given a perturbation $T(x) = T_0 + xE$ of $T_0$ such that $T(x)$ is also stochastic, how small must $x$ be in order for us to recover the indices of the direct summands of $T_0$?  We give a simple algorithm based on the orthogonal projection matrix onto the left or right singular subspace corresponding to the $k$ smallest singular values of $I - T(x)$ which allows for exact recovery all clusters when $x = O\left(\frac{\sigma_{n - k}}{||E||_2\sqrt{n_1}}\right)$ and approximate recovery of a single cluster when $x = O\left(\frac{\sigma_{n - k}}{||E||_2}\right)$, where $n_1$ is the size of the largest cluster and $\sigma_{n - k}$ the $(k + 1)$st smallest singular value of $T_0$.}

\section{Introduction}

This note, which is a companion to~\cite{cole2021clusters}, addresses the problem of clustering in Markov chains: given the transition matrix of a finite Markov chain, can we partition the states into subsets (``clusters'') with the property that the density of transitions within the same cluster is high, while the density of transitions between different clusters is low?  Clustering in networks, sometimes called \emph{community detection}, is a problem of central importance in the burgeoning field of data science; see~\cite{fortunato2010community} for survey.  In the context of Markov chains, clustering has numerous applications, including identifying metastable conformations of biomolecules~\cite{DHFS}, identifying neighbourhoods in urban traffic networks \cite{CKS}, and identifying well-connected regions in protein-protein interaction networks \cite{Azadetal}.

Algorithms for identifying clusters in Markov chains typically rely on the eigenvalues and eigenvectors~\cite{DHFS,DW,CKS,BCFKS} or singular values and singular vectors~\cite{fritzsche2008svd} of the associated transition matrix $T$.  In~\cite{cole2021clusters} an algorithm based on the singular values and associated left singular vectors of the \emph{Laplacian matrix} $I - T$ is presented.  The contribution of this note is that it combines the Laplacian-based clustering from~\cite{cole2021clusters} with projection techniques commonly used in the \emph{planted partition problem} (a much-studied clustering problem on random graphs)~\cite{mcsherry,vu,ColeFR,cole2019recovering}.

The setting we consider in this note is the same as that of~\cite{cole2021clusters}: we are given a perturbation $T(x) := T_0 + xE$ of a \emph{completely decoupled} transition matrix $T_0$, i.e.\ a direct sum of $k \geq 2$ irreducible stochastic matrices $T_1, \ldots, T_k$ such that $T(x)$ is stochastic for sufficiently small $x$.  Our goal is then to ``recover'' the clusters (i.e.\ blocks of indices) of $T_0$ given only $T(x)$, i.e.\ to determine exactly which pairs of indices belong to the same block of $T_0$ and which pairs belong to different blocks.  We seek to accomplish the following for sufficiently small $x$:
\begin{itemize}
\item	\textbf{Exact recovery of all clusters} (Section~\ref{sec:l2diff}).  We wish to determine \emph{exactly} the indices of \emph{each} block of $T_0$, up to a permutation of the blocks.  We present an algorithm which accomplishes  this when $x = O\left(\frac{\sigma_{n - k}}{||E||_2\sqrt{n_1}}\right)$, where $n_1$ is the size of the largest cluster and $\sigma_{n - k}$ the $(k + 1)$st smallest singular value of $T_0$.
\item	\textbf{Approximate recovery of one cluster} (Section~\ref{sec:approx}).  We wish to produce a set $\hat S$ with small symmetric difference with one of the clusters.  We present an algorithm which accomplishes this when $x = O\left(\frac{\sigma_{n - k}}{||E||_2}\right)$, assuming that all of the clusters are the same size.
\end{itemize}
Surprisingly, if the algorithm for approximate recovery of one cluster is adapted for exact recovery, it yields no improvement over the algorithm for exact recovery of all clusters.  It remains to be seen whether these guarantees are optimal.

\subsection{Notation}

We will use the following notation throughout this work:

\begin{itemize}
\item	$T_0 := T_1 \oplus \ldots \oplus T_k \in \R^{n \times n}$, where $T_i$ is an irreducible stochastic matrix with index set $S_i$.  Note that $S_1, \ldots, S_n$ are known as \emph{clusters} and partition the index set $[n]$.

\item	$S(i)$ -- the unique cluster containing index $i$.

\item	$n_1 \geq \ldots \geq n_k$ -- the orders of $T_1, \ldots, T_k$ arranged in nonincreasing order.  Note that $n_i$ is not necessarily the order of $T_i$. 

\item	$T(x) := T_0 + xE$, where $x \geq 0$ and $T(x)$ is stochastic for some $x > 0$.  Note that each row of $E$ must sum to 0.

\item	$e_1, \ldots, e_n$ -- standard basis vectors in $\R^n$.

\item	$I_n$ -- $n \times n$ identity matrix.  The subscript is omitted if it is clear from context.

\item	$\sigma_1(A) \geq \ldots \geq \sigma_{\min\{m, n\}}(A) \geq 0$ -- the singular values of a matrix $A \in \R^{m \times n}$.

\item	$P_k^{L}(A)$, $P_k^{R}(A)$ -- the orthogonal projection matrices onto the left and right singular subspaces corresponding to the \emph{smallest} $k$ singular values of a matrix $A$, respectively.

\item	$P^L(x) := P_k^L(I - T(x))$, $P^R(x) := P_k^R(I - T(x))$.
\end{itemize}
Note also that we sometimes refer to a Markov chain and its transition matrix interchangeably.

\section{Deviation in orthogonal projection matrices}

For a matrix $A \in \R^{m \times n}$, consider the linear subspaces $\mathbf U$ and $\mathbf V$ spanned by the left and right singular vectors corresponding to the smallest $k$ singular values of $A$.  Let $P_k^{L}(A)$ and $P_k^{R}(A)$ denote the orthogonal projection matrices onto $\mathbf U$ and $\mathbf V$, respectively.  If $A$ is symmetric, then these are both equal to the orthogonal projection matrix onto the subspace spanned by the eigenvalues corresponding to the $k$ smallest eigenvalues of $A$ in absolute value, which we denote by $P_k(A)$.  

\begin{lemma}\label{lemma:projdiff}
Let $A, B \in \R^{n \times n}$ be symmetric and $\beta > \alpha > 0$.  If the $k$ smallest eigenvalues in absolute value of both $A$ and $B$ are $\leq \alpha$, and all other eigenvalues of both $A$ and $B$ are $\geq \beta$ in absolute value, then
\[||P_k(A) - P_k(B)||_2 \leq \frac{2||A - B||_2}{\beta - \alpha}.\]
\end{lemma}

\begin{proof}[Proof sketch]
Apply the Cauchy integral formula, as in the proof of~\cite[Lemma~4]{cole2019recovering}.  Integrate over the rectangle in the complex plane whose left and right sides are on the lines $x = \pm\frac{\beta - \alpha}2$ and whose top and bottom sides are on the lines $y = \pm M$, and let $M \to \infty$.
\end{proof}

For a matrix $A \in \R^{n \times n}$, define $S(A) := 
\begin{bmatrix}
0	& A	\\
A^\top	&	0
\end{bmatrix}
\in \R^{2n \times 2n}$.  We will call this the \emph{symmetrization} of $A$.  Observe that the eigenvalues of $S(A)$ are the singular values of $A$ and their opposites.  More specifically, if $\sigma$ is a singular value of $A$ with left and right singular vectors $u$ and $v$, respectively, then $\sigma$ is an eigenvalue of $S(A)$ with eigenvector $
\begin{bmatrix}
u	\\
v	\\
\end{bmatrix}
$, and $-\sigma$ is an eigenvalue of $S(A)$ with eigenvector $
\begin{bmatrix}
u	\\
-v	\\
\end{bmatrix}
$.  Also note that $||S(A)||_2 = ||A||_2$, since the spectral norm of a symmetric matrix is the maximum absolute value among its eigenvalues.  Finally, it is not difficult to show that in fact $P_{2k}(S(A)) = P_k^L(A) \oplus P_k^R(A)$.

Thus, we get the following as an immediate corollary to Lemma~\ref{lemma:projdiff}:

\begin{corollary}\label{cor:asymmprojdiff}
Let $A, B \in \R^{n \times n}$ and $\beta > \alpha > 0$.  If the $k$ smallest singular values of both $A$ and $B$ are $\leq \alpha$, and all remaining singular values of both $A$ and $B$ are $\geq \beta$, then 
\[||P_k^L(A) - P_k^L(B)||_2,\ ||P_k^R(A) - P_k^R(B)||_2 \leq \frac{2||A - B||_2}{\beta - \alpha}.\]
\end{corollary}

\begin{proof}
From Lemma~\ref{lemma:projdiff} we immediately get
\[||P_{2k}(S(A)) - P_{2k}(S(B))||_2 \leq \frac{2||A - B||_2}{\beta - \alpha}.\]
The conclusion follows from the fact that
\begin{eqnarray*}
||P_{2k}(S(A)) - P_{2k}(S(B))||_2	& =	& ||(P_k^L(A) \oplus P_k^R(A)) - (P_k^L(B) \oplus P_k^R(B))||_2	\\
	& =	& ||(P_k^L(A) - P_k^L(B)) \oplus (P_k^R(A) - P_k^R(B))||_2	\\
	& =	& \max\{||P_k^L(A) - P_k^L(B)||_2, ||P_k^R(A) - P_k^R(B)||_2\}.\qedhere
\end{eqnarray*}
\end{proof}

Applying Corollary~\ref{cor:asymmprojdiff} to the difference of a direct sum of irreducible stochastic matrices and a perturbation thereof, we get the following:

\begin{theorem}\label{thm:markovprojdiff}
Let $T(x) = T_0 + x E \in \R^n$, where $T_0 = T_1 \oplus \ldots \oplus T_k$, $T_i$ is an irreducible stochastic matrix for $i = 1, \ldots k$, and $T(x)$ is stochastic for sufficiently small $x$.  Let $P^L(x) := P_k^L(I - T(x))$ and $P^R(x) := P_k^R(I - T(x))$.  Then
\begin{equation}\label{eqn:markovprojdiff}
||P^L(x) - P^L(0)||_2,\ ||P^R(x) - P^R(0)||_2 \leq \frac{2x||E||_2}{\sigma_{n - k}(I - T_0) - 2x||E||_2},
\end{equation}
provided that $2x||E||_2 < \sigma_{n - k}(I - T_0)$.

Furthermore, $P^L(0) = \bigoplus_{i = 1}^ku_iu_i^\top$ and $P^R(0) = \bigoplus_{i = 1}^kv_iv_i^\top$, where $u_i$ an $v_i$ are (respectively) left and right Perron vectors of $T_i$ with unit $\ell_2$-norm.
\end{theorem}

\begin{proof}
By Weyl's inequalities we have
\[\sigma_i(I - T(x)) \leq \sigma_i(I - T_0) + x||E||_2 = x||E||_2\]
for $ i > n - k$ and
\[\sigma_i(I - T(x)) \geq \sigma_i(I - T_0) - x||E||_2 \geq \sigma_{n - k}(I - T_0) - x||E||_2\]
for $i \leq n - k$.  Thus,~\eqref{eqn:markovprojdiff} follows by applying Corollary~\ref{cor:asymmprojdiff} with $\alpha = x||E||_2$ and $\beta = \sigma_{n - k}(I - T_0) - x||E||_2$.

Now observe that if we define $\tilde u_i \in \R^n$ by putting $u_i$ in the indices corresponding to $T_i$ and 0s elsewhere, then $\tilde u_1, \ldots, \tilde u_k$ is an orthonormal basis for the left null space of $I - T_0$; hence, 
\[P^L(0) = \sum_{i = 1}^k\tilde u_i\tilde u_i^\top = \bigoplus_{i = 1}^ku_iu_i^\top.\]
The conclusion about $P^R(0)$ follows similarly. 
\end{proof}

\section{Exact recovery of all clusters}\label{sec:l2diff}


Let us define
\begin{equation}\label{eqn:epsilondef}
\epsilon := \frac{2x||E||_2}{\sigma_{n - k}(I - T_0) - 2x||E||_2}.
\end{equation}
From~\eqref{eqn:markovprojdiff} we have $||(P^R(x) - P^R(0))e_j||_2 \leq \epsilon$ for $j = 1, \ldots n$, where $e_j \in \R^n$ is the $j$th standard basis vector.  In other words, the $j$th columns of $P^R(x)$ and $P^R(0)$ differ by at most $\epsilon$ in $\ell_2$-norm.

\begin{theorem}\label{thm:l2diff}
Let $T(x) = T_0 + x E \in \R^n$, where $T_0 = T_1 \oplus \ldots \oplus T_k$, $T_i$ is an irreducible stochastic matrix with index set $S_i$ for $i = 1, \ldots k$, and $T(x)$ is stochastic for sufficiently small $x$.  Let $P_k^R(x)$ be defined as in Theorem~\ref{thm:markovprojdiff}, and let $\epsilon$ be defined by~\eqref{eqn:epsilondef}.  For $j = 1, \ldots, n$, let $S(j)$ be the block $S_i$ containing $j$.  Then the following are true:
\begin{itemize}
\item	If $S(i) = S(j)$, then $||P_k^R(x)e_i - P_k^R(x)e_j||_2 \leq 2\epsilon$.  
\item	If $S(i) \neq S(j)$, then $||P_k^R(x)e_i - P_k^R(x)e_j||_2 \geq \sqrt{\frac1{|S(i)|} + \frac1{|S(j)|}} - 2\epsilon$.
\end{itemize}
\end{theorem}

\begin{proof}
For convenience, let us define $P = [p_1 ,\ldots, p_n] := P_k^R(0)$, $\hat P = [\hat p_1, \ldots, \hat p_n] := P_k^R(x)$.  Then by Theorem~\ref{thm:markovprojdiff} and~\eqref{eqn:epsilondef} we have $||P - \hat P||_2 \leq \epsilon$, hence $||p_j - \hat p_j||_2 = ||(P - \hat P)e_j||_2 \leq \epsilon$ for $j = 1, \ldots, n$.

Now let us compare columns of $\hat P$:
\[\hat p_i - \hat p_j = p_i - p_j + (\hat p_i - p_i) + (p_j - \hat p_j).\]
Hence
\begin{equation}\label{eqn:coldiff}
\Big|||\hat p_i - \hat p_j||_2 - ||p_i - p_j||_2\Big| \leq ||(\hat p_i - \hat p_j) - (p_i - p_j)||_2 = ||(\hat p_i - p_i) + (p_j - \hat p_j)||_2 \leq 2\epsilon.
\end{equation}

Now, using the normalized indicator vectors of the clusters as an orthonormal basis for the right null-space of $I - T_0$, Theorem~\ref{thm:markovprojdiff} gives
\begin{equation}\label{eqn:pi-pj}
||p_i - p_j||_2 = \left\{
\begin{array}{ll}
0	& \textrm{if } S(i) = S(j),	\\
\sqrt{\frac1{|S(i)|} + \frac1{|S(j)|}}	& \textrm{else.}
\end{array}
\right.
\end{equation}
Thus, if $S(i) = S(j)$, then~\eqref{eqn:coldiff} gives $||\hat p_i - \hat p_j||_2 \leq 2\epsilon$; otherwise, we get $||\hat p_i - \hat p_j||_2 \geq \sqrt{\frac1{|S(i)|} + \frac1{|S(j)|}} - 2\epsilon$.  This completes the proof.
\end{proof}

Theorem~\ref{thm:l2diff} suggests the following algorithm to recover the clusters, provided that $2\epsilon < \sqrt{\frac1{n_1} + \frac1{n_2}} - 2\epsilon$ (where $n_1 \geq \ldots \geq n_k$ are the cluster sizes): 
\begin{algorithm}[H]
\caption{Exact recovery of all clusters}\label{alg:exact}
For each pair of distinct indices $i, j$, label $i$ and $j$ as being in the same cluster iff.\ $||P^R(x)e_i - P^R(x)e_j||_2 \leq \tau$, and return the resulting partition of the indices.  
\end{algorithm}
By Theorem~\ref{thm:l2diff}, if $\epsilon < \frac14\sqrt{\frac1{n_1} + \frac1{n_2}}$, then Algorithm~\ref{alg:exact} with $\tau = 2\epsilon$ can distinguish between pairs of indices in the same cluster and pairs of indices in different clusters.  By~\eqref{eqn:epsilondef}, this is satisfied when $\displaystyle x = O\left(\frac{\sigma_{n - k}(I - T_0)}{||E||_2\sqrt{n_1}}\right)$.

Note that we arbitrarily choose to use $P^R(x)$ in Theorem~\ref{thm:l2diff} and Algorithm~\ref{alg:exact}; we could have just as easily used $P^L(x)$.  Moreover, $P^R(x)$ and $P^L(x)$ can be computed efficiently and reliably as follows.  Let $I - T(x) = U\Sigma V^\top$ be the singular value decomposition of $I - T(x)$, and let $U_k, V_k \in \R^{n \times k}$ be the last $k$ columns of $U$ and $V$, respectively.  Then $P^L(x) = U_kU_k^\top$, and $P^R(x) = V_kV_k^\top$.  Such SVD-based computations are numerically stable, and there are specialized methods that allow $U_k$ and $V_k$ to be computed without computing the entire SVD~\cite{vanHuff1,vanHuff2}. 

\section{Approximate recovery of one cluster}\label{sec:approx}

A slightly more refined analysis shows that with a looser restriction on $x$ we are able to approximately recover \emph{one} of the clusters.

Let us assume for simplicity that all clusters are of size $s := n / k$.  Let $S(j)$ be defined as in Theorem~\ref{thm:l2diff} and 
\[\hat S(j) := \left\{i : P^R(x)_{ij} \geq \frac1{2s}\right\}.\]
We will see that for some $j$ $\hat S(j)$, which can be determined empirically from $T(x)$, is a good approximation to $S(j)$.  

Observe that by Theorem~\ref{thm:markovprojdiff} we have
\[||P_k^R(x) - P_k^R(0)||_\F^2 = \sum_{j = 1}^n||(P_k^R(x) - P_k^R(0))e_j||_2^2 \leq 2k\epsilon^2.\]
Since the entries of $P_k^R(0)$ are either $0$ or $1 / s$, by averaging we get
\[\frac{|S(j) \triangle \hat S(j)|}{4s^2} \leq ||(P_k^R(x) - P_k^R(0))e_j||_2^2 \leq \frac{2k\epsilon^2}n\]
for some $j$, where $\triangle$ denotes symmetric difference.  Hence, $|S(j) \triangle \hat S(j)| \leq \frac{8\epsilon^2s^2k}n = 8\epsilon^2s$ for some $j$.  How do we identify such a ``good'' $j$?  

Arguing as in~\cite[Section~8.1]{cole2019recovering}, we can show that if $|\hat S(j)| \leq (1 + 8\epsilon^2)s$ and $||P_k^R(x)\1_{\hat S(j)}||_2 \geq (1 - \epsilon - O(\epsilon^2))\sqrt s$, then $|S(j) \cap \hat S(j)| \geq (1 - 3\epsilon)s$ and $|S(j) \triangle \hat S(j)| \leq 4\epsilon s$, and there will always exist such an $\hat S(j)$.  Thus, we can identify a $\hat S(j)$ with small symmetric difference with one of the blocks of $T_0$ by taking the one which maximizes $||P_k^R(x)\1_{\hat S(j)}||_2$ among those with $|\hat S(j)| \leq (1 + 8\epsilon^2)s$.  Hence, by~\eqref{eqn:epsilondef}, if $\displaystyle x = O\left(\frac{\sigma_{n - k}(I - T_0)}{||E||_2}\right)$ we can approximately recover \emph{one} of the clusters, i.e., construct a set of indices in which a large proportion come from a single cluster.  This yields the following algorithm:
\begin{algorithm}[H]
\caption{Approximate recovery of one cluster}\label{alg:approx}
Let $j$ be the index which maximizes $||P_k^R(x)\1_{\hat S(j)}||_2$ among those with $|\hat S(j)| \leq (1 + 8\epsilon^2)s$, and return $\hat S(j)$.
\end{algorithm}
Again, note that we could have used $P_k^L(x)$ in place of $P_k^R(x)$ in the above algorithm and analysis.

If we want exact recovery, then we need $\epsilon < \frac1{3s}$, hence $\displaystyle x = O\left(\frac{\sigma_{n - k}(I - T_0)k}{||E||_2n}\right)$.  This is no better than the guarantees in Section~\ref{sec:l2diff}; hence, the rounding trick used in this section is only useful for \emph{approximate} recovery of \emph{one} cluster.  It would be nice to show that it can be iterated to recover additional clusters, but it might be difficult to show that the errors don't accumulate too much as we iterate.

Note that this averaging trick \emph{does} give a better guarantee for exact recovery in the planted partition problem.  The reason is the randomness of the input: once we have an approximate cluster we can recover it exactly \emph{with high probability} through a simple error correcting procedure; see~\cite[Section~8.2]{cole2019recovering}.

\section{Determining the threshold empirically}

In Section~\ref{sec:l2diff} we use $\tau = 2\epsilon$ as the threshold in Algorithm~\ref{alg:exact}.  Thus, the algorithm is assumed to have access to the value of $\epsilon$.  Of course, this may not be the case in practice, as $\epsilon$ depends on $\sigma_{n - k}(I - T_0)$, $x$, and $||E||_2$.  While $\sigma_{n - k}(I - T_0)$ may be approximated with $\sigma_{n - k}(I - T(x))$ and $||E||_2$ assumed to be a constant (e.g.\ 1), knowing $x$ means we know exactly how small a perturbation $T(x)$ is from $T_0$, which may not be a realistic assumption.
There are several ways around this.

\subsection{If we know the two largest cluster sizes}

If $2\epsilon < \sqrt{\frac1{n_1} + \frac1{n_2}} - 2\epsilon$, then $2\epsilon < \frac12\sqrt{\frac1{n_1} + \frac1{n_2}} < \sqrt{\frac1{n_1} + \frac1{n_2}} - 2\epsilon$.  Hence, we can use $\tau = \frac12\sqrt{\frac1{n_1} + \frac1{n_2}}$ as the cutoff in Algorithm~\ref{alg:exact} instead of $\tau = 2\epsilon$.  But what if our algorithm doesn't even have access to $n_1$ and $n_2$?

\subsection{If no cluster sizes are known}

Let us call a positive number ``Small'' if it is $\leq 2\epsilon$, ``Large'' if it is $\geq \sqrt{\frac1{n_1} + \frac1{n_2}} - 2\epsilon$.  Our goal is to distinguish the Small values of $||P^R(x)e_i - P^R(x)e_j||_2$ from the Large ones, without actually knowing $\epsilon$, $n_1$, and $n_2$.  By Theorem~\ref{thm:l2diff}, this tells us exactly which pairs of indices are in the same cluster and which pairs are in different clusters.

Let $d_1 \geq \ldots \geq d_{n \choose 2}$ be the ${n \choose 2}$ values of $||P^R(x)e_i - P^R(x)e_j||_2$ for all $i \neq j$, arranged in nonincreasing order.  By Theorem~\ref{thm:l2diff} each $d_i$ is either Large or Small.  We want to find the unique $i$ such that $d_i$ is Large and $d_{i + 1}$ is Small.  We could just try using each $d_i$ as the cutoff $\tau$ in Algorithm~\ref{alg:exact}, construct the corresponding decoupled transition matrix $\hat T_i$ (whose blocks are sub-stochastic rather than stochastic), and picking the $i$ that minimizes $||\hat T_i - T(x)||$ for an appropriately chosen norm.

However, if we assume not only that $2\epsilon < \sqrt{\frac1{n_1} + \frac1{n_2}} - 2\epsilon$, but that in fact $4\epsilon < \sqrt{\frac1{n_1} + \frac1{n_2}} - 2\epsilon$, then we can narrow the search space to just $O(\log n)$ possibilities.  If $4\epsilon < \sqrt{\frac1{n_1} + \frac1{n_2}} - 2\epsilon$, then any Large number must be at least \emph{twice} any Small number.  Hence, if $d_i$ is Large and $d_{i + 1}$ is Small, then it must be the case that $d_i \geq 2d_{i + 1}$.  Thus, we can limit our search space to the indices $i$ such that $d_i \geq 2d_{i + 1}$.  We will call such an index $i$ a \emph{gap} index.

Now, let $m$ be the largest index such that $d_m$ is Large.  Then of course
\[d_m \geq \sqrt{\frac1{n_1} + \frac1{n_2}} - 2\epsilon > \frac12\sqrt{\frac1{n_1} + \frac1{n_2}}.\]
We can also easily show using~\eqref{eqn:coldiff} and~\eqref{eqn:pi-pj} that
\[d_1 \leq \sqrt{\frac1{n_k} + \frac1{n_{k - 1}}} + 2\epsilon < \sqrt{\frac1{n_k} + \frac1{n_{k - 1}}} + \frac12\sqrt{\frac1{n_1} + \frac1{n_2}}.\]
Now let $h$ be the number of gap indices $\leq m$.  Since $d_i$ is halved after every gap index, we have $d_m \leq \left(\frac12\right)^hd_1$.  Hence,
\begin{eqnarray*}
h 	& \leq	& \log_2\left(\frac{d_1}{d_m}\right) 	\\
	& < 	& \log_2\left(\frac{\sqrt{\frac1{n_k} + \frac1{n_{k - 1}}} + \frac12\sqrt{\frac1{n_1} + \frac1{n_2}}}{\frac12\sqrt{\frac1{n_1} + \frac1{n_2}}}\right) 	\\
	& =		& \log_2\left(\frac{\sqrt{\frac1{n_k} + \frac1{n_{k - 1}}}}{\frac12\sqrt{\frac1{n_1} + \frac1{n_2}}} + 1\right)	\\
	& \leq	& \log_2\left(\frac{\sqrt{2 / n_k}}{\frac12\sqrt{2/n_1}} + 1\right)	\\
	& =		& \log_2\left(2\sqrt{\frac{n_1}{n_k}} + 1\right)	\\
	& \leq	& \log_2(2\sqrt n + 1).
\end{eqnarray*}
Thus, if we start with $d_1$ and try the first $\lfloor\log_2(2\sqrt n + 1)\rfloor$ gap indices we find, $d_m$ is guaranteed to be among them.  Hence, we are able to empirically determine a threshold to use in Algorithm~\ref{alg:exact} (in place of $2\epsilon$) at the expense of a logarithmic factor increase in running time.  Observe that this factor may be improved if we happen to know an upper bound for $n_1 / n_k$.

We state the algorithm explicitly here:

\begin{algorithm}[H]
\caption{Determining $\tau$ empirically}
\begin{enumerate}
\item	Let $d_1 \geq \ldots \geq d_{n \choose 2}$ be the ${n \choose 2}$ values of $||P^R(x)e_i - P^R(x)e_j||_2$ for all $i \neq j$, arranged in nonincreasing order.
\item	For $i = 1, \ldots, {n \choose 2} - 1$, if $d_i \geq 2d_{i + 1}$, run Algorithm~\ref{alg:exact} with $\tau = d_{i + 1}$.  Let $\hat T_i$ be the decoupled sub-stochastic matrix induced by the output of Algorithm~\ref{alg:exact}; that is, $\hat T_i := \bigoplus_{j = 1}^l\hat T_{i, j}$, where $\hat T_{i, 1}, \ldots, \hat T_{i, l}$ are the principle submatrices of $T(x)$ induced by the parts of the partition returned by Algorithm~\ref{alg:exact}.\label{step:gap}
\item	Stop when $\lfloor\log_2(2\sqrt n + 1)\rfloor$ indices $i$ have been tried in Step~\ref{step:gap} and return the partition that produces the $\hat T_i$ which minimizes $||\hat T_i - T(x)||$ for an appropriately chosen matrix norm $||\cdot||$.
\end{enumerate}
\end{algorithm}

Note that in this section we still assume $k$ (the number of clusters) is known a priori.  If this is not the case, we can use a similar ``search and check'' procedure to determine $k$ empirically, although it increases the running time of Algorithm~\ref{alg:exact} by a factor of $n$.

\section{Concluding remarks}

It remains open whether the restrictions on $x$ which guarantee the success of Algorithms~\ref{alg:exact} and~\ref{alg:approx} are optimal.  Moreover, it is interesting that if one wishes to use Algorithm~\ref{alg:approx} for exact exact recovery, then the guarantee is no better than that of Algorithm~\ref{alg:exact}.  Thus, it is natural to ask whether there is an algorithm for \emph{exact} recovery of a \emph{single} cluster with better guarantees (i.e.\ a looser restriction on $x$) than Algorithm~\ref{alg:exact}.

\section*{Acknowledgement}

Many thanks to Steve Kirkland for his helpful comments.

\bibliographystyle{plain}
\bibliography{projrec_markov}

\end{document}